\documentclass[12pt,a4paper]{article} \usepackage[english]{babel}
\usepackage{amstext,amsfonts,amsmath,amssymb} \usepackage{amsthm}
\usepackage{color}
\parskip=\medskipamount
\theoremstyle{definition}
\newtheorem{definition}{Definition}
\newtheorem{example}{Example} 

\newtheorem{theorem}{Theorem}
\newtheorem{proposition}[theorem]{Proposition}
\newtheorem{lemma}[theorem]{Lemma}
\newtheorem{corollary}[theorem]{Corollary}

\def\pa{\partial}
\begin{document}

\title{Solvability of a Lie algebra of vector fields\\ implies their integrability
by quadratures}

\author{
J.F. Cari\~nena$^{\dagger\,a)}$,
F. Falceto$^{\ddagger\,b)}$,
J. Grabowski$^{\diamond\,c)}$,\\ \\
${}^{\dagger}$
   {\it Departamento de F\'{\i}sica Te\'orica and IUMA, Facultad de Ciencias} \\
   {\it Universidad de Zaragoza, 50009 Zaragoza, Spain}  \\
${}^{\ddagger}$
   {\it Departamento de F\'{\i}sica Te\'orica and BIFI, Facultad de Ciencias} \\
   {\it Universidad de Zaragoza, 50009 Zaragoza, Spain}  \\
${}^\diamond$
  {\it Polish Academy of Sciences, Institute of Mathematics,   }
  \\{\it  \'Sniadeckich 8,  00-656 Warsaw, Poland }
 }
\date{\today}
\maketitle


\begin{abstract}
We present a substantial generalisation of a classical result by Lie on integrability by quadratures. Namely, we prove that all vector fields in a finite-dimensional 
transitive and solvable Lie algebra of vector fields on a manifold can be integrated by quadratures.
\end{abstract}

\begin{quote}

{\it  Short title: Solvability implies integrability}

{\it MSC Classification:}
{\enskip}34A26, {\enskip}37J15, {\enskip}37J35, {\enskip}70H06
\end{quote}

{\vfill}

\footnoterule
{\noindent\small
$^{a)}${\it E-mail address:} {jfc@unizar.es } \\
$^{b)}${\it E-mail address:} {falceto@unizar.es }\\
$^{c)}${\it E-mail address:} {jagrab@impan.pl }}
\newpage

 \tableofcontents
\section{Introduction}

The property of integrability of a given {autonomous} system of differential equations  has been receiving   quite a lot of attention and it has been an active field of research along the last years
 because of its applications  in many branches of science, and in particular in physics.
 The exact meaning of integrability {is stated clearly in each specific field only}.
Loosely speaking, integrability refers to the possibility of finding the general solution of the system.
One can be interested only in some kind of  solutions, for instance, polynomial or rational ones, or expressible in terms of elementary functions.

>From the geometric point of view, systems of differential equations appear as associated {with} a particular choice of coordinates to vector fields on a manifold $M$: solutions of the system provide local integral curves of the associated vector field {in these coordinates}. Therefore, integrability of a vector field means that you can find in an algorithmic way its flow.

The existence  of additional compatible geometric structures may play a relevant role and it allows us to introduce other concepts of integrability. {For instance,} the notion of integrability is often
identified as \emph{complete integrability}  or \emph{Arnold--Liouville integrability} \cite{arnold}, but we can also consider \emph{generalised Liouville
integrability} {\cite{MF78,Z04}} or even \emph{non-Hamiltonian integrability} \cite{MF78}.

Our aim  is to extend the  study of  integrability in the absence of additional compatible structures developed in \cite{CFGR15}, and more specifically the classical problem of integrability by quadratures, {by using} modern tools of algebra and geometry. This means to study  under what conditions  you can determine the solutions (i.e. the flow of a given vector field $\Gamma$) by means of a finite number of algebraic operations (including inversion of functions) and quadratures of some functions.

{ An important remark is that integrability by quadratures refers to a given system of coordinates. For instance, it is well known that for any vector field $\Gamma\in\mathfrak{X}(M)$ which does not vanish at $p\in M$ there is a neighbourhood $U$ of $p$ and coordinates $Q^1,\dots,Q^n$ on $U$ in which $\Gamma$ is \emph{rectified}, i.e. takes the form $\partial/\partial{Q^1}$. The problem is that finding
this coordinate system is equivalent to integrating $\Gamma$, something that sometimes cannot be done {by quadratures}, when we start from a `bad' coordinate system. On the other hand, when additional data are available (in our cases this will be a certain Lie algebra $L$ of vector fields on $M$ containing $\Gamma$), then the autonomous system of ordinary differential equations $\dot x=\Gamma(x)$ can be solved by quadratures independently on our initial system of coordinates. In this case we say that \emph{$\Gamma$ is integrable by quadratures}.

A classical example is the celebrated result due to Lie, who established the following theorem {\cite{MF78}}:

\begin{theorem}\label{t-Lie}
If $n$ vector fields, $X_1$,\ldots,$X_n$, which are linearly independent at each point of  an open set  $U\subset\mathbb{R}^n$, span a solvable Lie algebra and satisfy $[X_1,X_i]=\lambda_i\, X_1$ with $\lambda_i\in \mathbb{R} $,
then $X_1$ is integrable by quadratures in $U$.
\end{theorem}
}
A {more general} integrability criterium we proposed in \cite{CFGR15} is based on the existence of a  $n$-dimensional Lie algebra  $L$   of vector fields in a $n$-dimensional manifold $M$ containing  the given vector field,
 very much in the approach started by Lie, and the construction of a decreasing sequence of nested Lie subalgebras of $L$ such that all of them contain the dynamical vector field $\Gamma$. If in some step the resulting Lie algebra is Abelian,
 we can obtain  the general solution with one more quadrature.

 Solvability of $L$ was proved to be a necessary conditions for integrability in the sense we proposed.  Moreover,  we
 proved that if the  solvable Lie algebra admits an Abelian ideal $A$, then any vector field of $A$ is integrable by quadratures (this fact being a direct generalisation of Lie result), and finally in the particular case of $L$ being a nilpotent Lie algebra { (for a general form of a transitive nilpotent Lie algebra of vector fields consult \cite{JG90})}, every vector field of $L$ was proved to be integrable.
A concept of \emph{distributional integrability} we introduced later on, allows for a much larger class of examples and potential applications.

A different result is due to Kozlov \cite{K05}. It can be formulated as follows.

\begin{theorem} Let vector fields, $X_1$,\ldots,$X_n$, be linearly independent at each point of  an open set  $U\subset\mathbb{R}^n$ and span a Lie algebra $L$ such that the corresponding operators of the adjoint representation ad$_{X_i}=[X_i,\cdot]$ have a common triangular form
\begin{equation}\label{triang}
[X_i,X_j]={\sum_{k=1}^i } C_{ij}\,^k\, X_k, \qquad C_{ij}\,^k\in \mathbb{R}\,.
\end{equation}
Then, all the vector fields $X_i$, $i=1,\dots, n$, are integrable by quadratures.
\end{theorem}

It seems that the belief of the author was that any solvable Lie algebra  admits a basis satisfying (\ref{triang}); see also \cite{K13}. This is valid for complex  Lie algebras but is generally false in the real case (see the next section). Moreover, the important restriction here  is that  the dimension $n$ of the   Lie algebra $L$ coincides with that of the manifold $M$, which {is a strong restriction and} essentially means that the manifold we deal with is practically a $n$-dimensional Lie group.

In this paper we present a substantial generalisation of the above results by Lie and Kozlov on integrability by quadratures. Namely, we are able to drop the assumption that the dimension of the Lie algebra coincides with the dimension of the manifold and prove that all vector fields in a finite-dimensional transitive and solvable Lie algebra of vector fields on a manifold can be integrated by quadratures. This substantially enlarges the number of examples. Recall that a Lie algebra of vector fields on a manifold $M$ is \emph{transitive} if it just spans the tangent bundle $TM$. It seems that it is a necessary assumption, as a deeper description of non-transitive Lie algebras of vector fields are generally out of our reach.
Note that in many cases transitive Lie algebras of vector fields, including those of infinite-dimension, determine the manifold \cite{JG78,S74}.

The paper is organised as follows: Section 2 is devoted to classical examples of solvable non-autonomous systems we will use in the proof of our results. In Section 3, we
introduce the notation and recall basic definitions and well-known results
of the theory of solvable Lie algebras. In particular, we consider some sequences of Lie algebras and ideals to be used in the sequel.
In Section 4, we develop a Lie-like formalism for a transitive finite-dimensional real Lie algebra $L$ of vector fields on a  differentiable manifold $M$ to integrate by quadratures some associated foliations.
The main result is proven in Section 5 which ends with an illustrative example not covered neither by the classical Lie's nor the Kozlov's result.

{ \section{Integrability by quadratures of non-autonomous systems}
We will start with classical examples of integrability by quadratures of non-autonomous systems which we will need in the sequel.

Consider first the non-autonomous inhomogeneous linear differential  equation in dimension one,}
\begin{equation}
\dot x=c_0(t)+ c_1(t)\,x, \label{ileq}
\end{equation}
which is well known to be integrable in terms of two  quadratures:
\begin{equation}
x(t)=\exp\left(\int_0^t c_1(t')\, dt'\right)\left[x_0+\int_0^{t}\exp\left(-\int_0^{t'}c_1(t'')\,
 dt''\right)c_0(t')\,dt'\right].\label{sileq}
\end{equation}
Another example is given by the  nonautonomous system of  differential  equations
\begin{equation}
\dot x^i=\sum_{j=1}^nH^i\,_j\,x^j+b^i(t), \quad i=1,\ldots ,n,\qquad \label{ilseq}
\end{equation}
where $H^i\,_j$ are real numbers. Then, the solution starting from the point $\boldsymbol{x}_0$ is given by
\begin{equation}
\boldsymbol{x}(t)=\exp\left(Ht\right)\left[\boldsymbol{x}_0+\int_0^t \exp\left(-Ht'\right) \, \boldsymbol{b}(t') dt'\right].\label{ssieq}
\end{equation}
These two examples will play an important role in the proof of the main theorem in this paper.

{ \section{Solvable Lie algebras}

Recall that the derived algebra of a  Lie algebra $(\mathfrak {g},[\cdot,\cdot ])$  is the subalgebra $\mathfrak {g}^{1}$ of $\mathfrak {g}$, defined by
$\mathfrak {g}^{1}=[{\mathfrak {g}},{\mathfrak {g}}]$, while the derived series is the sequence
 of Lie subalgebras defined by $\mathfrak {g}^{0}=\mathfrak {g}$ and
\begin{equation}
\mathfrak {g}^{k+1}=[\mathfrak {g}^{k},\mathfrak {g}^{k}],\quad k\in\mathbb{N}.
\end{equation}
Such a sequence satisfies $\mathfrak {g}^{k+1}\subset \mathfrak {g}^{k}$, i.e.
$$  {\mathfrak {g}}\supset [{\mathfrak {g}},{\mathfrak {g}}]\supset [[{\mathfrak {g}},{\mathfrak {g}}],[{\mathfrak {g}},{\mathfrak {g}}]]\supset [[[{\mathfrak {g}},{\mathfrak {g}}],[{\mathfrak {g}},{\mathfrak {g}}]],[[{\mathfrak {g}},{\mathfrak {g}}],[{\mathfrak {g}},{\mathfrak {g}}]]]\supset \cdots,
$$
and the Lie algebra $\mathfrak {g}$ is said to be \emph{solvable} if the derived series eventually arrives at the zero subalgebra, i.e. there exists {the} smallest
natural number $m$ such that
$\mathfrak {g}^{m+1}=\{0\}$ or, in other words, $\mathfrak {g}^{m}$ is Abelian.
}

On the other hand,  such a solvable Lie algebra  $\mathfrak {g}$ always admits a nontrivial ideal. For instance, remark that   $[\mathfrak {g},\mathfrak {g}]$ is strictly contained in $\mathfrak {g}$
and that any codimension one linear subspace $S$  containing  $[\mathfrak {g},\mathfrak {g}]$ is an ideal.
Note also that, using the Jacobi identity and mathematical induction on the index $k$, one can easily check that for {any} $k\in\mathbb{N}$, $[\mathfrak {g},\mathfrak {g}^{k}]\subset  \mathfrak {g}^{k}$,  i.e.
$\mathfrak {g}^{k}$ is an ideal of $\mathfrak {g}$.
Moreover, $\mathfrak {g}^{m}$ is an Abelian ideal.
For later purposes we will be interested in Abelian ideals
of dimension not greater than two.

Another important Lie theorem {(see e.g. \cite{H72, S90})} ensures that every { finite-dimensional} representation of a solvable Lie algebra { over an algebraically closed field}
has an eigenvector common to all
the operators of the representation. If we consider the adjoint
representation, the theorem implies that any { finite-dimensional} complex, solvable Lie algebra has a one dimensional ideal.
This result is not true in general for real Lie algebras, as  the following
example shows.

\begin{example} The real Lie algebra
$\mathfrak{e}(2)$  of the {Euclidean} group in two dimensions
is given by
$\mathfrak{e}(2)=\langle X_1,X_2,J\rangle$
with Lie brackets
$$[X_1,X_2]=0,\quad [X_1,J]=X_2,\quad [X_2,J]=-X_1.$$
It is solvable as its derived algebra $\mathfrak{d}=[\mathfrak{e}(2),
\mathfrak{e}(2)]=\langle X_1,X_2\rangle$
is Abelian.
However, its only non-trivial ideal is precisely
$\mathfrak{d}$ which has dimension 2.
\end{example}

We will be interested in the real framework and in this case,
using again the mentioned  Lie theorem \cite{H72, S90}, we can prove that every real solvable
Lie algebra has an Abelian ideal of dimension 1 or 2, which will be just
enough for our purposes. {In fact, we can consider the complexified Lie algebra and its adjoint representation for which we can use the standard Lie theorem.} As there is  a common complex eigenvalue  $\lambda$,  its complex  conjugate $\bar\lambda$ is another eigenvalue {(the characteristic polynomial is real)} and the complex subspace spanned by the eigenvectors of both contains {an at most} two-dimensional real subspace which is an Abelian ideal of the given solvable Lie algebra. This proves the following:
\begin{proposition}\label{solvid}
Any solvable finite-dimensional real Lie algebra $L$ contains {an Abelian} Lie ideal $A$ of dimension $1$ or $2$.
\end{proposition}

Another well-known series of nested algebras is the central series;
it is defined by $\mathfrak {g}_{0}=\mathfrak {g}$ and
\begin{equation}
\mathfrak {g}_{k+1}=[\mathfrak {g},\mathfrak {g}_{k}],\quad k=0,1,\dots.
\end{equation}

It is clear that $\mathfrak {g}_{k+1}\subset\mathfrak {g}_{k}$
and $\mathfrak {g}^{k}\subset\mathfrak {g}_{k}$.

If the central series stabilises at zero, the Lie algebra is called \emph{nilpotent}.
The preceding property ensures that nilpotent Lie algebras are
always solvable. On the other hand, if an algebra is nilpotent or solvable,
then so are its subalgebras {and quotients}.

A third series we will be interested in requires to pick up first an
Abelian subalgebra
$\mathfrak{a}\subset \mathfrak{g}$.
It is defined by
\begin{equation}\label{defgai}
\mathfrak{g}_{\mathfrak{a}}^0=\mathfrak{g},\qquad
\mathfrak {g}_{\mathfrak{a}}^{k+1}=[\mathfrak {g}_{\mathfrak{a}}^{k},\mathfrak {g}_{\mathfrak{a}}^{k}]+\mathfrak{a},\quad k=0,1,\dots.
\end{equation}

One can easily show the following

\begin{proposition} The preceding sequence is a decreasing sequence of  Lie subalgebras, i.e. $\mathfrak{g}_{\mathfrak{a}}^{i+1}\subset \mathfrak{g}_{\mathfrak{a}}^{i}$ for any index $i$, and moreover
$\mathfrak{g}_{\mathfrak{a}}^{i+1}$ is an ideal in $\mathfrak{g}_{\mathfrak{a}}^{i}$.
\end{proposition}
\begin{proof}
First we apply the mathematical induction to show that
$\mathfrak{g}_{\mathfrak{a}}^{i+1}\subset \mathfrak{g}_{\mathfrak{a}}^{i}$ {for any $i$}.

Certainly $\mathfrak{g}_{\mathfrak{a}}^{1}\subset \mathfrak{g}_{\mathfrak{a}}^{0}=\mathfrak{g}$. Recalling the definition (\ref{defgai}) of the sequence $\mathfrak{g}_{\mathfrak{a}}^i$
and using the induction hypothesis,
($\mathfrak{g}_{\mathfrak{a}}^{i}\subset\mathfrak{g}_{\mathfrak{a}}^{i-1}$ for $i\leq k$), we get
$$\mathfrak{g}_{\mathfrak{a}}^{k+1}=[\mathfrak{g}_{\mathfrak{a}}^{k},\mathfrak{g}_{\mathfrak{a}}^{k}]+
\mathfrak{a}\subset[\mathfrak{g}_{\mathfrak{a}}^{k-1},\mathfrak{g}_{\mathfrak{a}}^{k-1}]+\mathfrak{a}=
\mathfrak{g}_{\mathfrak{a}}^{k}\,,$$
{that ends the induction.}
{Now, for each index $i$, the inclusions
$$[\mathfrak{g}_{\mathfrak{a}}^{i+1},\mathfrak{g}_{\mathfrak{a}}^{i}]
\subset[\mathfrak{g}_{\mathfrak{a}}^{i},\mathfrak{g}_{\mathfrak{a}}^{i}]\subset \mathfrak{g}_{\mathfrak{a}}^{i+1},$$
show that $\mathfrak{g}_{\mathfrak{a}}^{i+1}$ is an ideal in $\mathfrak{g}_{\mathfrak{a}}^{i}$ (thus a Lie subalgebra}.
\end{proof}

The relative position of the series is given by the following
relation
$$\mathfrak{g}^{k}+\mathfrak{a}
\subset
\mathfrak{g}_{\mathfrak{a}}^{k}
\subset
\mathfrak{g}_{k}+\mathfrak{a}
$$
that can be easily proved by induction.
One also has $\mathfrak{g}_{\mathfrak{a}'}^{k}
\subset\mathfrak{g}_{\mathfrak{a}}^{k}$ for $\mathfrak{a}'\subset\mathfrak{a}$.

\begin{definition}
If the series $\mathfrak{g}_{\mathfrak{a}}^{k}$ stabilises at $\mathfrak{a}$, then
we say that $\mathfrak{g}$ is \emph{$\mathfrak{a}$-solvable}.
\end{definition}

The preceding relations imply that if  $\mathfrak{g}$ is
$\mathfrak{a}$-solvable for some Abelian subalgebra $\mathfrak{a}$,
then it is solvable, and if $\mathfrak{g}$ is nilpotent, then
it is  $\mathfrak{a}$-solvable for any $\mathfrak{a}\subset \mathfrak{g}$.
Also, $\mathfrak{a}$-solvability implies $\mathfrak{a}'$-solvability for
any $\mathfrak{a'}\subset\mathfrak{a}$.

One can easily show the following
{ \begin{proposition}
If $\mathfrak{a}\vartriangleleft \,
\mathfrak{g}$ is an Abelian Lie ideal,
we have
\begin{equation}\label{id-a}
\mathfrak{g}_{\mathfrak{a}}^k=\mathfrak{g}^k+\mathfrak{a}\,.
\end{equation}
In particular, the sequence $(\mathfrak{g}_{\mathfrak{a}}^k)$ is a decreasing sequence of  Lie ideals in $\mathfrak{g}$.
\end{proposition}
\begin{proof}
An easy induction shows that $\mathfrak{g}^k\subset\mathfrak{g}_{\mathfrak{a}}^k$,
so clearly also $\mathfrak{g}^k+\mathfrak{a}\subset\mathfrak{g}^k_{\mathfrak{a}}$.

To show the other inclusion, we can observe that
if $\mathfrak{g}_{\mathfrak{a}}^{k}\subset\mathfrak{g}^{k}+\mathfrak{a}$, then
$$\mathfrak{g}_{\mathfrak{a}}^{k+1}\subset[\mathfrak{g}^{k}+\mathfrak{a},\mathfrak{g}^{k}+\mathfrak{a}]
\subset \mathfrak{g}^{k+1}+[\mathfrak{g}^{k},\mathfrak{a}]+[\mathfrak{a},\mathfrak{a}]\subset
\mathfrak{g}_{\mathfrak{a}}^{k+1}+\mathfrak{a}\,,$$
since $\mathfrak{a}$ is a Lie ideal in $\mathfrak{g}$.
Of course, $\mathfrak{g}^{k}+\mathfrak{a}$ is a Lie ideal in $\mathfrak{g}$ as the sum of two ideals.
\end{proof}
The following corollary is an immediate consequence of the above proposition.
\begin{corollary}\label{asolvable}
For any Abelian Lie ideal $\mathfrak{a}
\vartriangleleft \,
\mathfrak{g}$, the Lie algebra $\mathfrak{g}$ is
$\mathfrak{a}$-solvable if and only if it is solvable.
\end{corollary}

\section{Generalised Lie algorithm for integration.}

In \cite{CFGR15} we studied the integration by quadratures of an
element of a Lie algebra of point-wise independent vector fields.
In this section we are going to apply the same ideas to a more general
setup in which the vector fields in the Lie algebra do not need to be pointwise
independent. This allows us to consider group actions
that do not need to be locally free and in particular to approach
the integration of vector fields around stationary points.

Let $L$ be a transitive finite-dimensional real subalgebra of vector fields on a $n$-dimensional differentiable manifold $M$, $L\subset \mathfrak{X}(M)$ \cite{D12,P01,S74}, and
consider a vector field $\Gamma \in L$. We will simplify slightly
the notation and we will denote by $L_\Gamma^i$ the Lie subalgebra
$L_{\langle \Gamma\rangle} ^i$ in the series associated to the subalgebra $\mathfrak{a}=
\langle \Gamma\rangle\subset L$ introduced in (\ref{defgai}), that is
\begin{equation}\label{defLGi}
L_\Gamma^0=L,\quad L_\Gamma^{i}=[L_\Gamma^{i-1},L_\Gamma^{i-1}]+\langle\Gamma\rangle,\ i\in {\mathbb{N}}
\end{equation}
Next, we consider the generalised distribution spanned by
$L_\Gamma^{i}$, to be denoted $\mathcal{D}_\Gamma^i$.  It is obviously
(Frobenius) integrable, because  $L_\Gamma^{i}$ is a Lie algebra. {With some abuse of notation, with $\mathcal{D}_\Gamma^i$ we will denote also the corresponding module of sections of this distribution.}
If we denote by $M^i$ one of the leaves of  $\mathcal{D}_\Gamma^i$ of dimension $r_\Gamma^i$
we have that upon restriction to $M^i$,  $\mathcal{D}_\Gamma^{i+1}\subset\mathfrak{X}(M^i)$.
One can prove that it is a regular distribution.
\begin{proposition}\label{regular}
For every $k$, $\mathcal{D}_\Gamma^k\subset\mathfrak{X}(M^{k-1})$
is a regular distribution and
therefore if $M^k$ is one of {its} leaves,
then $L_\Gamma^k\subset\mathfrak{X}(M^{k})$ is a transitive
real Lie subalgebra of vector fields in $M^k$.
\end{proposition}
\begin{proof}
We prove the proposition by induction. Assume that $\mathcal{D}_\Gamma^i$ is
transitive in $M^i$ for $i<k$, which certainly is true for $k=1$.
Now, as $L_\Gamma^{k}$ is an ideal in $L_\Gamma^{k-1}$,
we have {that $[\mathcal{D}_\Gamma^k, L_\Gamma^{k-1}]\subset\mathcal{D}_\Gamma^k$,}
which implies that $\mathcal{D}_\Gamma^k$ is invariant under the flows of the vector fields of
$L_\Gamma^{k-1}$. But {as these  act} transitively on $M^{k-1}$, {as assumed} in the
 induction hypothesis, then the rank $r_\Gamma^k$ of
$\mathcal{D}_\Gamma^k\subset\mathfrak{X}(M^{k-1})$ is constant.
Hence, the flows of $L_\Gamma^{k}$ act transitively on any of its leaves
$M^k$.
\end{proof}

The algorithm to integrate the differential equation 
induced by $\Gamma$ proceeds by iteratively applying an elementary
step that is described in the following proposition.

\begin{proposition}\label{leaves}
The  leaves $M^k$ of the foliation $\mathcal{D}_\Gamma^k\subset
\mathfrak{X}(M^{k-1})$ can be determined  by quadratures.
\end{proposition}
\begin{proof}
The proof is local. Consider a point $p\in M^{k-1}$ and  a maximal set
of vector fields which are linearly independent at $p$,
$B=\{X_1,X_2,\dots,X_{r_\Gamma^k}\}\subset L_\Gamma^{k}$.
Due to the assumed transitivity of $L_\Gamma^{k}$,
they form a local basis of
the {module} $\mathcal{D}_\Gamma^{k}$. Let us supplement $B$ with more vector
fields $\{Y_1,\dots,Y_{r_\Gamma^{k-1}-r_\Gamma^k}\}\in L_\Gamma^{k-1}$, so that their values at each point  $p\in M^{k-1}$
form a basis of $T_pM^{k-1}$.
Now, it is easy to check that the 1-forms
$\alpha^s\in\bigwedge^1(M^{k-1})$, with $s=1,\dots,r_\Gamma^{k-1}-r_\Gamma^k$,
locally defined by the conditions of vanishing on $L_\Gamma^{k}$
and  such that $\alpha^s(Y_r)=\delta^{s}_r$,
are closed.
{We can then integrate them by quadratures}  to obtain
locally defined  functions $Q^s$, so that $\alpha^s=\mathrm{d}\,Q^s$.
Finally, the leaves $M^k$ are simply the intersections of the
level sets of the $Q^s, s=1,\dots,r_\Gamma^{k-1}-r_\Gamma^k$.
\end{proof}

{\bf Remark:} Note that, while the proposition is true in general,
it is only useful when $r_\Gamma^{k-1}>r_\Gamma^k$. The case in which the two ranks
are equal is trivial.

The previous results can be combined to prove the following
\begin{theorem}
Let $L$ be a finite-dimension real transitive solvable subalgebra of vector fields on a
differentiable manifold $M$, $L\subset \mathfrak{X}(M)$, $\Gamma \in L$.
If $L$ is $\langle\Gamma\rangle$-solvable,
then the system of differential equations determining the integral curves
of $\Gamma$ is integrable by quadratures
\end{theorem}
\begin{proof}
The proof is a simple iteration of the previous procedure.
We determine { by quadratures}
at every step the integral leaves of
{ $\mathcal{D}_\Gamma^k$}. If we end up with {$\mathcal{D}_\Gamma^m=\langle\Gamma\rangle$ for an index $m$},
the corresponding leaves $M^m$ are one-dimensional and the problem of finding
the integral curves of $\Gamma$ reduces to solving an autonomous
differential equation in one variable, which can always be achieved by
one quadrature.
\end{proof}

The previous results can be trivially extended to Abelian Lie subalgebras
\begin{corollary} If $A\subset L$ is an Abelian Lie subalgebra
and  $L$ is $A$-solvable, then
$(M,L,\Gamma)$ is integrable by quadratures
for any vector field $\Gamma\in A$.
\end{corollary}
\begin{proof}
  It is enough to consider that for $\Gamma\in A$
  we have $L^k_\Gamma\subset L^k_A$, then if $L$ is $A$-solvable,
    i.e. $L^m_A=A$, then we have
    $L^m_\Gamma\subset A$ and as $A$ is Abelian
    $L^{m+1}_\Gamma=\langle\Gamma\rangle$. Therefore $L$ is
    $\langle\Gamma\rangle$-solvable and therefore can be integrated
    by quadratures.
\end{proof}

In the previous paragraphs we have proved the integrability by quadratures
of certain Abelian subalgebras of vector fields in a transitive Lie algebra.
A related question is that of {\it straightening out} or {\it rectifying}  the subalgebra
in a sense made more precise in the following definition:

\begin{definition}
We say that an Abelian subalgebra of vector fields $A$ is straightened out,
or rectified, in  an open set $U$, 
if we can find local coordinates $(Q^1,\dots,Q^n)$ in $U$ such that the set 
$\{\partial_{Q^1},\dots,\partial_{Q^r}\}\subset A$ 
and it generates the same distribution as $A$.
\end{definition}

This definition generalises the similar concept for a vector field. 
Of course, the obstruction for the existence of such coordinates
relies in the fact that $A$, in general, does not need to be of constant rank.

To better understand the situation, we will consider the following example
 \begin{example}\label{januszex}
Let $L$ be the solvable and transitive Lie algebra
of vector fields on $\mathbb{R}^2$
spanned by $\partial_x,\partial_y,x\partial_x,y\partial_y,y^2\partial_x,y\partial_x$:
$$L=\langle \partial_x,\partial_y,x\partial_x,y\partial_y,y^2\partial_x,y\partial_x\rangle\,.
$$
We take two different Abelian subalgebras $A_1=\langle \partial_x\rangle$
and $A_2=\langle y\partial_x\rangle$. Their associated descending series
are quite similar
\begin{eqnarray*}
  L^{1}_{A_1}&=&
  L^{1}_{A_2}=
  \langle\partial_x,\partial_y,y^2\partial_x,y\partial_x\rangle,\\
L^{2}_{A_1}&=&
L^{2}_{A_2}=
\langle\partial_x,y\partial_x\rangle,\\
L^{3}_{A_1}&=&A_1,\
L^{3}_{A_2}=A_2.
\end{eqnarray*}
The fact that both series {stabilize} at an Abelian subalgebra
implies the Lie integrability of the corresponding vector fields.
One may also notice that, while $A_2$ is of {non-constant} rank in
$\mathbb{R}^2$, it is indeed of constant rank in every of the leaves
of $\mathcal{D}_{A_2}^2$ that are given by $y=$const., as proposition
\ref{regular} assures.
However, while $A_1$ can be straightened out,
actually it is already rectified in coordinates $x,y$,  in the second
case  this cannot be done because $A_2$ is of {non-constant} rank.

A third interesting instance is when the Abelian subalgebra is 
$A_3=A_1+A_2$. In this case the dimension of the subalgebra is  two while its  
rank is  1, and then they  do not coincide.
Of course $A_3$ can be rectified,
in fact, according to our definition, it is straightened in $x,y$ coordinates.

\end{example}

In view of the previous example, the natural question arises whether there is
an algebraic way to distinguish between both situations.
Actually, the key difference among the subalgebras in the example
is that while $A_1$ and $A_3$  are Lie ideals of  $L$,  $A_2$ is not, in fact
$[L,A_2]=A_1+A_2$.  We will show below that this is the general case
and, indeed,  it is always possible to rectify an Abelian ideal of a solvable Lie algebra of vector fields.

Given a $n$-dimensional differentiable manifold $M$, let $L\subset
\mathfrak{X}(M)$ be a transitive finite-dimensional Lie algebra
as above.
Consider an Abelian ideal $A\vartriangleleft L$ and the descending series $L^i_A$ as defined in the previous section,
that is
\begin{equation}\label{defLGi}
L_A^0=L,\quad L_A^{i+1}=[L_A^{i},L_A^{i}]+A=L^{i+1}+A\,,\ i\in {\mathbb{N}}\,.
\end{equation}
Next we consider the generalised distribution $\mathcal{D}_A^i$ spanned by
$L_A^{i}$. It is obviously
(Frobenius) integrable, because  $L_A^{i}$ is a Lie algebra.
One can prove that it is actually a regular distribution.
\begin{proposition}\label{coordinates}
Let $A$ be an Abelian ideal of a transitive Lie algebra 
of vector fields $L$. For every $k$, the distribution $\mathcal{D}_A^k$
is a regular and involutive distribution on $M$, and then generates a foliation
$\mathcal{F}_A^k$.
In particular, if $M^k$ is one of {its} leaves,
then $L_A^k$ restricted to $M^{k}$ is a transitive
real Lie subalgebra of vector fields on $M^k$.
\end{proposition}
\begin{proof}
As $L_A^{k}$ is an ideal in $L$ and the latter is transitive, the compositions of local diffeomorphisms generated by flows of $L$ act transitively on $M$ and preserve the distribution $D_A^k$ which is therefore of constant rank.
\end{proof}

As distributions $\mathcal{D}_A^k$ are regular, they have constant rank, which will be denoted with $r_A^k$. Of course, $r_A^{k+1}\le r_A^k$.
A crucial step in showing our result on the integrability by quadratures 
is the following.
\begin{theorem}\label{t1}
Given an Abelian ideal $A$ in a transitive Lie algbera $L$ as before, 
for each $p\in M$, there is a neighbourhood $U$ of $p$ and a coordinate system $(Q)=(Q^1,\dots,Q^n)$ therein, which can be obtained by quadratures from any given one, such that $(Q)$  locally determines the series of foliations $\mathcal{F}_A^k$, $k=0,1,2,\dots$, i.e.
\begin{equation}\label{fol}
\text{the leaves of}\quad\mathcal{F}_A^k\cap U\quad \text{are the level sets of}\quad Q^{r_A^k+1},\dots,Q^n\,.
\end{equation}
Moreover, if $L^m_A=A$, then these coordinates can be chosen so that $\partial_{Q^1},\dots,\partial_{Q^{r_A^m}}$ belong to $A$.
\end{theorem}
\begin{proof}
The proof is inductive and local.

Suppose $Q^{r_A^k+1},\dots,Q^n$ are chosen so that (\ref{fol}) is satisfied and let
 $X_{r_A^{k+1}+1}, \ldots, X_{r_{k}}\in L_A^{k}$ be chosen so that they span
 a vector subspace of $\mathcal{D}_A^{k}(p)$ complementary to $\mathcal{D}_A^{k+1}(p)$, i.e.
 $$\langle X_{r_A^{k+1}+1}(p), \ldots,X_{r_A^{k}}(p)\rangle \oplus  \mathcal{D}_A^{k+1}(p)=\mathcal{D}^{k}(p).$$
{These vector fields span therefore a subdistribution of $\mathcal{D}_A^{k}$ complementary to $\mathcal{D}_A^{k+1}$ in a neighbourhood of $p$.}
Of course, one understands that if $r_A^{k+1}=r_A^k$, then we do not need to choose anything (as $\mathcal{F}_A^k=\mathcal{F}_A^{k+1}$) and we pass to the next step. On any leaf $M^k$ of the foliation $\mathcal{F}_A^k$ in a neighbourhood of $p$ there are uniquely defined 1-forms
$\alpha^{r_A^{k+1}+1}, \ldots, \alpha^{r_{k}}$ which vanish on $\mathcal{D}_A^{k+1}$ and satisfy
$$
 \alpha^i(X_j)=\delta _{ij}, \qquad i,j=r_A^{k+1}+1, \ldots,r_A^k\,.
 $$
These forms are closed, so {locally} of the form $\alpha^i=\mathrm{d} Q^i_{M^k}$ for some (local) functions $Q^i_{M^k}$ on $M^k$, which can be obtained by quadratures (by integrating $\alpha^i$'s). 
The functions are defined modulo constants, but  if  $\mathcal{F}_A^{k}$ is a regular foliation
in $M$, they could be `synchronized' along all $M^k$ by  assuming that they vanish
on an arbitrarily chosen smooth section of the local projection $M\supset U\to(U/\mathcal{F}_A^k\cap U)$.
In this way we get, by quadratures, functions $Q^{r_A^{k+1}+1},\dots,Q^{r_A^k}$, defined in a neighbourhood $U$ of $p$, 
which vanish on $\mathcal{D}^{k+1}_A$. So  their level sets,
together with  those of $Q^{r_A^k+1},\dots, Q^n$, 
determine (locally) $\mathcal{F}_A^{k+1}$.
Remark that coordinates in $U$ with the required properties exist if $A$ is an ideal of $L$
and therefore $\mathcal{D}_A^k$ is a regular distribution, but not necessarily in the more general case
of $A$ being a subalgebra. This is the crucial step in which this proof differs from that of
Proposition \ref{leaves} in which this regularity was not ensured.

Finally, the coordinates have been chosen so that
$$
X_j(Q^i)= \alpha^i(X_j)=\delta _{ij}, \qquad i,j=r_A^{k+1}+1, \ldots,r_A^k\,,
$$
which, together with (\ref{fol})  and provided  that $L_A^m=A$, 
imply $X_i=\partial_{Q^i}$ for $i=1,\dots,r_m$ 
\end{proof}

The immediate consequence of the preceding proposition is the 
following corollary, that explains example \ref{januszex}
and is crucial 
for the results in the next section.
\begin{corollary}\label{straight}
Any Abelian ideal of a transitive finite-dimensional solvable Lie 
algebra of vector fields, can be straightened out by quadratures. 
\end{corollary}
\begin{proof}
As shown in Corollary \ref{asolvable}, if $A\vartriangleleft L$ is an Abelian ideal of a 
solvable Lie algebra, then $L$ is $A$-solvable, which means that 
$L_A^m=A$ for some $m$. Hence, the coordinates
determined by quadratures in  Proposition \ref{coordinates} 
rectify $A$.
\end{proof}

\section{Solvability implies integrability}

Here we will prove our main result which can be formulated as follows.
\begin{theorem}
If $L$ is a finite-dimensional solvable and transitive real Lie algebra of vector fields on a manifold $M$, then each vector field $\Gamma\in L$
is integrable by quadratures.
\end{theorem}
\begin{proof}
We shall use induction on the dimension $n$ of the manifold $M$, but as the considerations are local, we can as well assume that $M=\mathbb{R}^n$.
The case $n=0$ is trivial, so assume that $n\ge 1$ and let us
pick up an Abelian ideal $A\subset L$ of dimension one or two, whose existence is granted
for real solvable finite-dimensional Lie algebras (Proposition \ref{solvid}).
Due to the fact that $L$ is transitive, we know that the distribution
$\mathcal{D}_A$ spanned by $A$ is regular, say of rank $r\le 2$. As it is also involutive, it generates a foliation $\mathcal{F}_A$. Moreover, one can obtain by quadratures a coordinate system $Q^1,\dots,Q^n$ such that $\mathcal{D}_A$ is generated by $\pa_{Q^1},\dots,\pa_{Q^r}\in A$ and leaves of $\mathcal{F}_A$ are the level sets of the functions $Q^{r+1},\dots,Q^n$ 
(Theorem \ref{t1} and Corollary \ref{straight}).

We will first consider the case in which the dimension of the Abelian Lie algebra $A$
coincides with the dimension of the integral leaves of the foliation $\mathcal{F}_A$, i.e. $\operatorname{dim}(A)=r$. Then, we can view locally $M$ as a local principal bundle with
the structure group $A$ viewed as an Abelian Lie group acting on $M$ and with the
(local) fibration $\pi:M\rightarrow B$ of $M$ on the space $B=M/A$ of $A$-orbits. The latter can be identified as the leaves of the foliation $\mathcal{F}_A$.

By fixing a {local} section $\sigma:B\rightarrow M$ we can locally describe
the points of $M$ by pairs $(\mathbf{v},b)\in A\times B$ through the diffeomorphism
$\xi:M\rightarrow A\times B$ given by $p=\Phi_{\mathbf{v}}(1,\sigma(b))$, where
$\Phi_{\mathbf{v}}(t,p_0)$ is the integral curve of the vector field
$\mathbf{v}$ at time
$t$ with the initial condition $p_0$ at $t=0$.

Now, we can use $\xi$ to transport vector fields {from }$M$ to $A\times B${;}
for instance, an element $\mathbf{v}\in A$
may be also viewed as a vertical vector field on $M$,
$\mathbf{v}\in\mathfrak{X}(M)$, then the image of $\mathbf{v}$ in
$A\times B$, i.e. $\xi_*\mathbf{v}\in\mathfrak{X}(A\times B)$,
is the directional derivative
in the  direction $\mathbf{v}$ in the vector space $A$.

The dynamical vector field $\Gamma$ acts linearly on the ideal $A$
by the commutator, namely
$$ [\Gamma,{\mathbf{v}}]=H{\mathbf{v}},$$
{where $H$ is a linear map  $H:A\rightarrow A$.}
Moreover, $\Gamma$ is projectable by $\pi_*$ onto the base manifold $B$.
{We denote} its projection $\bar\Gamma=\pi_*\Gamma\in\mathfrak{X}(B)$.
The system of differential equations determining the integral curves of $\Gamma$ can therefore be written
{as follows:}
\begin{eqnarray}\label{sys}
\dot {\mathbf{v}} &=& H{\mathbf{v}} + {\mathbf{w}}(b)\\
\dot b &=& \bar\Gamma(b)\nonumber
\end{eqnarray}
for some function ${\mathbf{w}}:B\rightarrow A$.

The projection $\pi_*$ is a homomorphism of the Lie algebra of vector fields $L$ onto
a Lie algebra $\bar L$ of vector fields on $B$. Clearly, $\bar L$ is solvable and transitive, so by the inductive assumption, $\bar\Gamma$ (thus the equation $\dot b=\bar\Gamma(b)$)  is integrable by quadratures.
If we plug its solution $b(t)$ into the first equation, we are left with
$$\dot {\mathbf{v}} = H{\mathbf{v}} + {\mathbf{w}}(b(t))\,,$$
i.e. {an inhomogenous linear} equation with constant coefficients which, as remarked in
Section 2, can be solved by quadratures.
To see it in coordinates, we first conclude that
$[\partial_{Q^i},\Gamma]=\sum_{j=1}^r h^j_i\partial_{Q^j}$ implies that
$$\Gamma=\sum_{j=1}^r\left( \sum_{i=1}^r
h^j_iQ^i+w^j(Q^{r+1},\dots,Q^n)\right)\partial_{Q^j}+\bar\Gamma\,,
$$
where $h^j_i\in\mathbb{R}$, and $w^j$ as well as the vector field $\bar\Gamma=\sum_{s=r+1}^n\gamma^s(Q^{r+1},\dots,Q^n)\partial_{Q^s}$ depend on coordinates
$Q^{r+1},\dots,Q^n$ only. This leads to the system (\ref{sys}) which in coordinates reads
\begin{eqnarray}\label{sys1}
\dot Q^j &=& \sum_{i=1}^r h^j_iQ^i + w^j(Q^{r+1},\dots,Q^n)\,,\quad j=1,\dots,r\,,\\
\dot Q^s &=& \gamma^s(Q^{r+1},\dots,Q^n)\,,\quad s=r+1,\dots,n\,.\label{sys2}
\end{eqnarray}
Solving (\ref{sys2}) by the inductive assumption, we end up with
$$\dot Q^j=\sum_{i=1}^r h^j_iQ^i +  w^j(Q^{r+1}(t),\dots,Q^n(t))\,,\quad j=1,\dots,r\,,
$$
which can be integrated by quadratures.

Now, we should {still} consider the possibility that the dimension of $A$ is two,
but {the dimension of the   integral leaves of the foliation $\mathcal{D}_A$ is one}. In this case, we chose a one-dimensional subspace $A_1\subset A$,
whose generator $X_1$ spans $\mathcal{D}_A$. As we already know, $X_1$ can be integrated
by quadratures and can be taken as $\partial_{Q^1}$ in our system of coordinates. We again construct a fibration $\pi:M\rightarrow B$, with $B$ being the manifold of leaves of $\mathcal{F}_A$.
Projectability of $\Gamma$ is again guaranteed, {and we denote }
$\bar\Gamma\in\mathfrak{X}(B)$ the vector field $\pi_*\Gamma$.
The crucial difference from the previous case is that for ${\mathbf{v}}\in A_1$ {(recall that $A_1$ is assumed to be one-dimensional)} we have
$[\Gamma,{\mathbf{v}}]\in A$ and therefore $[\Gamma,{\mathbf{v}}]=f(b){\mathbf{v}}$, {where $f$ can now depend on coordinates $b$.}
The equations of motion read
\begin{eqnarray}
\dot {\mathbf{v}} &=& f(b){\mathbf{v}} + {\mathbf{w}}(b)\\
\dot b &=& \bar\Gamma(b)
\end{eqnarray}
for some function ${\mathbf{w}}:B\rightarrow A_1$.

{Once  again, if we} solve the second set of equations, we are left with
$$\dot {\mathbf{v}} = f(b(t)){\mathbf{v}} + {\mathbf{w}}(b(t))\,,$$
which is a linear differential equation with time-dependent coefficients.
As remarked in Section 2, due to the fact that this is a single equation in one variable,  it can be integrated by two quadratures (in contrast with the more general case
of a system of linear differential equations with time dependent coefficients,
that cannot be integrated in general).

In coordinates: $\Gamma$ must be of the form
$$\Gamma=\left(f(Q^2,\dots,Q^n)Q^1+w(Q^2,\dots,Q^n)\right)\partial_{Q^1}+\sum_{s=2}^n\gamma_s(Q^{2},
\dots,Q^n)\partial_{Q^s}\,.
$$
We can first solve
$$\dot Q^s=\gamma_s(Q^{2},\dots,Q^n)\,,\quad s=2,\dots,n\,,$$
and so reduce to
$$\dot Q^1=f(t)Q^1+w(t)$$
which also can be solved by quadratures.
\end{proof}

\begin{example}
Consider the Lie algebra of vector fields in
$\mathbb{R}^2$ spanned by
$$X_1=\partial_x,\quad X_2=y\partial_x,\quad J=xy\partial_x+(1+y^2)\partial_y.$$
The Lie algebra $L$ is isomorphic to $\mathfrak{e}(2)$ and
$A=\langle\partial_x, y\partial_x\rangle$ is its only non trivial ideal.

If we take $\Gamma=J$ as the dynamical vector field,
we immediately see that the  Lie's procedure cannot be applied, as $J$ is not
an element of any commutative ideal in $L$. Also the mentioned Kozlov's result is not applicable, since the algebra is not triangular and 
the vector fileds are not independent at every point.
To integrate by quadratures, we may use
the procedure described above. It applies as follows.

Take $A_1=\langle\partial_x\rangle$ and coordinate $y$ for the base
manifold $B$.
The equation for the coordinate $x$ in the fibre is
$$\dot x=xy,$$
while the differential equation corresponding to the
projection $\bar\Gamma$ of the dynamical vector field on $B$ is
$$\dot y=1+y^2,$$
that can be immediately integrated to give $y(t)=y_0+\tan t$.
Substituting into the equation in the fibre, we get
$$\dot x=(y_0+\tan t)x\,,$$
whose solution can be expressed as
$x(t)=x_0\exp(y_0t)/\cos t$.
\end{example}

\section*{Acknowledgments}

We are deeply indebted to M.F. Ra\~nada, coworker in the previous paper, for his essential contribution to the discussions at first stages of this one.
This work was partially  supported by the research projects MTM2015-64166-C2-1-P,  FPA2015-65745-P (MINECO/FEDER)  and DGA E24/1, E24/2 (DGA, Zaragoza). The research of
 J.~Grabowski was funded by the  Polish National Science Centre grant under the contract number DEC-2012/06/A/ST1/00256.



\begin{thebibliography}{99}

\bibitem{arnold} V.I. Arnold, {\sl Mathematical methods of classical mechanics}. Graduate Texts in Mathematics {\bf 60}, second edition, Springer,  1989.

\bibitem{MF78} A. Mishchenko and A. Fomenko,
 {\it Generalised Liouville method of integration of Hamiltonian systems}, Funct. Anal. Appl.   \textbf{12}, 113--121  (1978).

\bibitem{Z04}  N.T. Zung, {\it Torus Actions and Integrable Systems}, arXiv:0407455

 \bibitem{CFGR15}  J.F. Cari\~nena, M. Falceto, J. Grabowski and M.F. Ra\~nada,
 {\it Geometry of Lie integrability by quadratures},
 J.  Phys. A:Math. Theor. {\bf 48},  215206 (2015)
 \bibitem{JG78} J. Grabowski,
 {\it Isomorphisms and ideals of the Lie algebras of vector   fields},   
 Invent.  Math. \textbf{50}, 13--33 (1978).
 
 \bibitem{JG90} J. Grabowski,
 {\it Remarks on nilpotent Lie algebras of vector fields},
J. Reine Angew. Math. \textbf{406}, 1--4 (1990).

  \bibitem{K05}  V.V. Kozlov, {\it Remarks on a Lie Theorem on the Integrability of Differential Equations in Closed Form},
Differential Equations {\bf 41},  588--590 (2005).
\bibitem{K13}  V.V. Kozlov, {\it The Euler-Jacobi-Lie integrability theorem},  Regul. Chaotic Dyn. {\bf 18}, 329--343 (2013).

\bibitem{H72} J.E. Humphreys, {\sl Introduction to Lie algebras and Representation Theory}, Graduate Texts in Mathematics {\bf 9}. Springer--Verlag, New York 1972.

\bibitem{S90} H. Samelson, {\sl  Notes on Lie Algebras}, Universitext, Springer, 1990

\bibitem{D12} J. Draisma, {\it Transitive Lie Algebras of Vector Fields: An Overview}, Qual. Theory Dyn. Syst. {\bf 11}, 39--60 (2012)

\bibitem{P01} G. Post, {\it  On the structure of transitively differential algebras},
J. Lie Theory {\bf 11},  111--128 (2001)

\bibitem{S74}H.J. Sussmann,  {\it  An Extension of a Theorem of Nagano on Transitive Lie Algebras},
Procs. Amer. Math.  Soc. {\bf 45},   349--356 (1974)

 \end{thebibliography}
\end{document}